\documentclass[12 pt]{article}
\usepackage{latexsym,amssymb,amsmath}
\usepackage{authblk}
\usepackage{amsthm}
\usepackage{hyperref}
\usepackage{xcolor}

\usepackage{fullpage}

\numberwithin{equation}{section}

\newtheorem{theorem}{Theorem}[section]
\newtheorem{lemma}[theorem]{Lemma}
\newtheorem{proposition}[theorem]{Proposition}
\newtheorem{corollary}[theorem]{Corollary}

\theoremstyle{definition}
\newtheorem{definition}[theorem]{Definition}

\newtheorem{example}[theorem]{Example}

\def\F{\Bbb F}

\def\({\left(}
\def\){\right)}

\newcommand{\rmv}[1]{}

\begin{document}


\title{Affine Cartesian codes with complementary duals}

\author{Hiram H. L\'opez \thanks{The first author was partially supported by
CONACyT, CVU No. 268999 project ``Network Codes", and Universidad Aut\'onoma de
Aguascalientes. }}

\author{Felice Manganiello} 

\author{Gretchen L. Matthews\thanks{The second and third authors were
    partially supported by the National Science Foundation under grant
    DMS-1547399.}}

\affil{Clemson University}

\date{}

\maketitle

\begin{abstract}
  A linear code $C$ with the property that $C \cap C^{\perp} = \{0 \}$
  is said to be a linear complementary dual, or LCD, code. In this
  paper, we consider generalized affine Cartesian codes which are
  LCD. Generalized affine Cartesian codes arise naturally as the duals
  of affine Cartesian codes in the same way that generalized
  Reed-Solomon codes arise as duals of Reed-Solomon codes. Generalized
  affine Cartesian codes are evaluation codes constructed by
  evaluating multivariate polynomials of bounded degree at points in
  $m$-dimensional Cartesian set over a finite field $K$ and scaling
  the coordinates. The LCD property depends on the scalars
  used. Because Reed-Solomon codes are a special case, we obtain a
  characterization of those generalized Reed-Solomon codes which are
  LCD along with the more
  general result for generalized affine Cartesian codes. 
  {\em MSC2010 classification codes}: 11T71; 94B27; 14G50

\end{abstract}
\section{Introduction}

In this paper, we consider linear codes which are linear complementary dual (LCD), a property introduced by Massey in 1992 \cite{JMassey}. An LCD code is a linear code that has only the zero word in common with its dual. Recall that a linear code $C$ is an $K$-subspace of $K^n$, where $K$ is a finite field. Given such a code $C$, its dual is $C^{\perp}:= \left\{ w \in K: w \cdot c=0 \ \forall c \in C \right\}$. Hence, if $C \subseteq K^n$ is LCD, then $C \cap C^{\perp} = \left\{ 0 \right\}$ and $C \oplus C^{\perp}=K^n$; of course, if $K$ is not finite and instead has characteristic $0$, then this naturally holds. 

In 2015, Carlet and Guilley \cite{Carlet} demonstrated how LCD codes provide countermeasures to side-channel and fault-injection attacks. They use the observation (made by Massey) that a generator matrix $G$ of an LCD code $C$ has the property that $GG^T$ is nonsingular; certainly, the same holds for a parity-check matrix $H$ of an LCD code $C$, meaning $HH^T$ is nonsingular. Suppose that data $x \in K^n$ is masked as $z:=x+e$ where $e \in K^n$. Since $C \oplus C^{\perp}=K^n$, there exists $(x',y) \in K^k \times K^{n-k}$ with $$z=x'G+yH.$$  Then 
$$zG^T(GG^T)^{-1}= x'GG^T(GG^T)^{-1}+\underbrace{yHG^T(GG^T)^{-1}}_{0}=x'$$
and 
$$zH^T(HH^T)^{-1}=\underbrace{x'GH^T(HH^T)^{-1}}_{0}+yHH^T(HH^T)^{-1}=y.$$ According to Carlet and Guilley, the countermeasure is $(d-1)^{th}$ degree secure where $d$ is the minimum distance of $C$, and the greater the degree of the countermeasure, the harder it is to pass a successful SCA. To consider a fault injection attack, suppose $z$ is modified into $z+\epsilon$ where $\epsilon \in K^n$. 
Then $\epsilon=eG+fH$ for some $(e,f) \in K^k \times K^{n-k}$. Detection amounts to distinguishing $z$ from $z+\epsilon$. We have that $$z+\epsilon=(x'+e)G+(y+f)H.$$ Then
$$
(z+\epsilon)H^T(HH^T)^{-1}=(x'+e)GH^T(HH^T)^{-1}+(y+f)HH^T(HH^T)^{-1}=y+f.
$$
Notice that $z+\epsilon=y$ if and only if $f=0$ if and only if $\epsilon \in C$. Thus, fault not detected if $\epsilon \in C$. If $wt(\epsilon)<d$, where $d$ is the minimum distance of $C$, then fault is detected. Both of these applications motivate the need for LCD codes $C$ with large minimum distance.

Recently, it was shown that every linear code over $\F_q$ with $q>3$ is equivalent to an LCD code, as demonstrated by 
Carlet, Mesnager, Tang, Qi, and Pellikaan \cite{Carlet_equiv}. However, explicit constructions of LCD codes remain elusive. There have been results on the characterizations of LCD codes from particular families. Among them, there are some results for algebraic geometry codes, a particular family of evaluation code \cite{AG_LCD}, cyclic codes \cite{cyclic}, quasi-cyclic codes \cite{quasicyclic}, and generalized Reed-Solomon codes \cite{ChenLiu}. 

In this paper, we consider LCD codes which are a special type of evaluation code, called a generalized affine Cartesian code. Generalized Reed-Solomon codes are a special case. Our results on generalized Reed-Solomon codes differ from those in \cite{ChenLiu} in that we provide a characterization of which generalized Reed-Solomon codes are LCD and give explicit constructions, as opposed to determining the existence of such codes with a particular set of parameters; our results apply to codes over fields of any characteristic (as opposed to odd prime powers) as well. 

A generalized affine Cartesian code is defined as follows. Let $K:=\mathbb{F}_q$ be a finite field with $q$ elements, and let $A_1,\ldots,A_m$ be  non-empty subsets of $K.$ Set $K^*:=K \setminus \{ 0 \}$. Define the {\it
Cartesian product set\/} 
$$
\mathcal{A}:=A_1\times\cdots\times A_m\subseteq K^{m}
.$$
 Let $S:=K[X_1,\ldots,X_m]$ 
be a polynomial ring, let $\pmb{a_1},\ldots,\pmb{a_n}$ be the points of $\mathcal{A}$,
and let $S_{< k}$ be 
the $K$-vector space of all polynomials of $S$ of
degree less than $k,$ where $k$ is a non-negative integer. Fix $v_1, \dots, v_n \in K^*$,
\rmv{Fix $n$ non-zero elements $v_{\pmb{a_1}},\ldots,v_{\pmb{a}_n}$ of the field $K$}
and  set $\pmb{v} =  \left( v_1, \dots, v_n \right)$. 
The {\it evaluation map\/} 
\begin{equation*}
\begin{array}{llll}
{\rm ev}_k\colon &S_{< k}\ & \to & K^{|\mathcal{A}|}\\
&f &\mapsto &  \left(v_1 f(\pmb{a_1}),\ldots,v_n f(\pmb{a_n})\right)
\end{array}
\end{equation*}
\rmv{
\begin{equation*}
{\rm ev}_k\colon S_{\leq k}\longrightarrow K^{|\mathcal{A}|},\ \ \ \ \ 
f\mapsto \left(v_{\pmb{a}_1}f(\pmb{a}_1),\ldots,v_{\pmb{a}_n}f(\pmb{a}_n)\right),
\end{equation*}}
defines a linear map of
$K$-vector spaces. The image of ${\rm ev}_k$, denoted by $C_{k}(\mathcal{A},\pmb{v})$,
defines a { linear code}, called the {\it generalized affine Cartesian evaluation code\/}
({\it Cartesian code} for short) of degree $k$ associated to $\mathcal{A}$ and $\pmb{v}.$ The {\it dimension\/} and the {\it length\/} of $C_{k}(\mathcal{A},\pmb{v})$ 
are given by $\dim_K C_{k}(\mathcal{A},\pmb{v})$ (dimension as $K$-vector space)
and  $|{\mathcal{A}}|,$  respectively.
The {\it minimum
distance\/} of $C_{k}(\mathcal{A},\pmb{v})$  is given by 
$$
d(C_k(\mathcal{A},\pmb{v})):=\min\{\|{\rm ev}_k(f)\|
\colon {\rm ev}_k(f)\neq 0; f\in S_{\leq k}\},
$$
where $\|{\rm ev}_k(f)\|$ is the number of non-zero
entries of ${\rm ev}_k(f)$. Generalized affine Cartesian codes arise naturally as duals of affine Cartesian codes; this is seen  in the computation by Beelen and Datta \cite{BeelenDatta}, though the codes are not mentioned by name. In this paper, we investigate them more fully. 

Cartesian codes are special types of affine Reed-Muller codes in the sense of 
\cite[p.~37]{tsfasman} and a type of affine variety codes, which were defined in \cite{fl}.
Cartesian codes are a generalization of $q$-ary Reed-Muller codes,
which are Cartesian codes with $A_1=\ldots=A_m=K.$

Cartesian codes have been studied in different works when $\pmb{v}=\pmb{1},$
the vector of ones: they appeared first time
in \cite{Geil} and then independently in \cite{lopez-villa}.
In \cite{Geil}, the authors study the basic parameters of Cartesian codes, they determine
optimal weights for the case when $\mathcal{A}$ is the product of two sets, and then
present two list decoding algorithms. In \cite{lopez-villa} the authors study
the vanishing ideal $I(\mathcal{A})$ and then, using commutative algebra tools,
for instance regularity, degree, Hilbert function,
the authors determine
the basic parameters of Cartesian codes in terms of the sizes of the components
of the Cartesian set.
In \cite{carvalho}, the author shows some results on higher
Hamming weights of Cartesian codes and gives a different proof for
the minimum distance using the concepts of
Gr\"obner basis and footprint of an ideal. In \cite{carvalho2} the authors
find several values for the second least weight of codewords, also known as the
next-to-minimal Hamming weight. In \cite{BeelenDatta} the authors find the
generalized Hamming weights and the dual of Cartesian codes.

\rmv{
It should be mentioned that we do not know of any efficient decoding
algorithm for Cartesian codes.  The reader is referred to \cite[Chapter~9]{CLO1},
\cite{joyner-decoding,van-lint} and the references there for some
available decoding algorithms for some families of linear codes.}

\rmv{
{\bf Introduction to Hermitian codes}

{\bf Write definition of dual}
A {\it linear code with a complementary dual} (or an LCD code) is defined
to be a linear code $C$ whose dual code $C^\perp$ satisfies
$C\cap C^\perp=\left\{\pmb{0}\right\}$. LCD codes were introduced by J. Massey
at \cite{JMassey}.

{\bf Why are LCD codes important?.}

{\bf Results about LCD codes?}}

This paper is organized as follows. Section \ref{gen} details properties of generalized affine Cartesian codes, hereafter referred to as Cartesian codes. In Section \ref{LCD}, we provide a characterization of Cartesian codes which are LCD. Examples are found in Section \ref{ex}, and our results are summarized in Section \ref{conclusion}. 

For all unexplained terminology and additional information, we refer to 
\cite{CLO1,Eisen,monalg} for commutative algebra and the theory of Hilbert functions,
\cite{MacWilliams-Sloane,tsfasman} for the theory of linear codes, and
\cite{carvalho3,carvalho4,delsarte-goethals-macwilliams,GRT,lachaud,sorensen,tochimani}
for other families of evaluation codes, including several 
variations of Reed-Muller codes and projective versions
of the Cartesian codes.

\section{Properties of generalized affine Cartesian codes} \label{gen}

In what follows, $n_i:=|A_i|$, the cardinality of $A_i$ for $i=1,\ldots,m$. An important characteristic for Cartesian codes and evaluation codes in general
is the fact that we can use commutative algebra methods to
study them. The reason is because the kernel of the evaluation map
${\rm ev}_k,$ is precisely $S_{< k}\cap I(\mathcal{A})$, where $I(\mathcal{A})$
is the {\it vanishing ideal\/} of $\mathcal{A}$ consisting of 
all polynomials of $S$ that vanish on ${\mathcal{A}}$.
Thus, the algebra of $S/\left(S_{< k}\cap I(\mathcal{A})\right)$
is related to the basic parameters of $C_{k}(\mathcal{A},\pmb{v})$.
Observe the kernel of the evaluation map ${\rm ev}_k$ is independent of $\pmb{v}$
because every entry of $\pmb{v}$ is non-zero. In fact, the vanishing ideal of $\mathcal{A}=A_1\times\cdots\times A_m$ is given by
$
I(\mathcal{A})=
\left(\prod_{a_1 \in A_1} (X_1 - a_1),\ldots,\prod_{a_m \in A_m} (X_m - a_m)\right)
$ \cite[Lemma 2.3]{lopez-villa}.

\rmv{
\begin{lemma}\cite[Lemma 2.3]{lopez-villa}\label{01-29-17}
The vanishing ideal of $\mathcal{A}=A_1\times\cdots\times A_m$ is given by
\[
I(\mathcal{A})=
\left(\prod_{a_1 \in A_1} (X_1 - a_1),\ldots,\prod_{a_m \in A_m} (X_m - a_m)\right).
\]
\end{lemma}}

A {\it monomial matrix} is a square matrix with exactly one nonzero entry in
each row and column. A monomial matrix $M$ can be written either in the form $D P$
or the form $P D$, where $D$ is a diagonal matrix with nonzero entries at the
diagonal and $P$ is a permutation matrix.
Let $C_1$ and $C_2$ be codes of the same length over the field $K$, and let $G_1$ be
a generator matrix for $C_1.$ Then $C_1$ and $C_2$ are monomially equivalent provided
there is a monomial matrix $M$ over the same field $K$ so that $G_1M$ is a generator
matrix of $C_2$. Monomially equivalent codes
have the same length, dimension and minimum distance. For more properties of monomially
equivalent codes, see \cite{huf-pless} and references there.

Using properties of monomially equivalent codes along with \cite[Theorems \ 3.1 and \ 3.8]{lopez-villa}, we have the following result.

\begin{theorem}\label{01-28-17}
Let $C_{k}(\mathcal{A},\pmb{v})$ be a Cartesian code. \begin{enumerate}
\item The length of $C_{k}(\mathcal{A},\pmb{v})$ is $n=n_1\cdots n_m.$
\item The dimension of $C_{k}(\mathcal{A},\pmb{v})$ is
  $n_1\cdots n_m$ (i.e.\ ${\rm ev}_k$ is surjective) if
  $k-1 \geq \sum_{i = 1}^m (n_i - 1)$, and \begin{equation*}
    \dim(C_{k}(\mathcal{A},\pmb{v})) = \sum_{j=0}^m (-1)^j \sum_{1
      \leq i_1 < \cdots < i_j \leq m} \binom{m + k - n_{i_1} - \cdots
      - n_{i_j}}{k - n_{i_1} - \cdots - n_{i_j}}
  \end{equation*}
  otherwise. 
\item  The minimum distance  of
  $C_{k}(\mathcal{A},\pmb{v})$ is 1 if
  $k-1 \geq \sum_{i = 1}^m (n_i - 1)$, and for
  $0 \leq k-1 < \sum_{i = 1}^m (n_i - 1)$ we have
  \[d(C_{k}(\mathcal{A},\pmb{v})) = (n_{s + 1} - \ell ) \prod_{i = s
      + 2}^m n_i, \] where $s$ and $\ell$ are uniquely defined by
  $k-1 = \sum_{i = 1}^s (n_i - 1) + \ell$ with
  $0 \leq \ell < n_{s + 1} - 1$ (if $s + 1 = m$ we understand that
  $\prod_{i = s + 2}^m n_i = 1$, and if $k-1 < n_1 - 1$ then we set
  $s = 0$ and $\ell = k$).
\end{enumerate}

\end{theorem}

\rmv{
\begin{proof}
As $C_{k}(\mathcal{A},\pmb{v})$ is monomially equivalent to $C_{k}(\mathcal{A},\pmb{1}),$
we have the result by
\cite[Theorem\ 3.1 and Theorem\ 3.8]{lopez-villa}.
\end{proof}}

In light of Theorem \ref{01-28-17}, from now on we assume that $k-1< \sum_{i = 1}^m (n_i - 1)$.

The rest of this section is devoted to prove that the dual of the Cartesian code
$C_{k}(\mathcal{A},\pmb{v})$ is
$C_{k^\prime}(\mathcal{A},\pmb{v}^\prime),$ where
$k^\prime:=\sum_{i = 1}^m (n_i - 1)-k+1$
and $\pmb{v}^\prime$ is as described below.
The dual of the Cartesian code
$C_{k}(\mathcal{A},\pmb{1}),$ the case when $\pmb{v}$ is the vector of ones, was previously found
in \cite{BeelenDatta}.

Given a positive integer $\ell,$ we define $\left[\ell\right]:=\left\{1,\ldots,\ell\right\}.$
Let $\prec$ be the {\it graded-lexicographic order} on the set of monomials of $S.$
This order is defined in the following way:
$X_1^{t_1}\cdots X_m^{t_m}\prec X_1^{s_1}\cdots X_m^{s_m}$ if and only if
$\sum_{i=1}^m t_i<\sum_{i=1}^m s_i$ or $\sum_{i=1}^m t_i=\sum_{i=1}^m s_i$
and the leftmost nonzero entry in $(s_1-t_1,\ldots,s_m-t_m)$ is positive.
From now on we fix the order $\prec$.
Denote the variables $X_1,\ldots,X_m$ by $\pmb{X}.$
For each $i\in[m],$ define the polynomial
\begin{equation}\label{03-24-18}
L_i(X_i):=\prod_{\substack{a_{i}\in A_i}}\left(X_i-a_{i}\right).
\end{equation}
Then, according to \cite[Lemma 2.3]{lopez-villa},
and \cite[Proposition 4]{CLO1}, $\left\{
L_1(X_1),\ldots,L_m(X_m)
\right\}$
is a Gr\"obner basis of $I(\mathcal{A}).$

Notice that for evaluation purposes we can assume that
$\deg_{X_i} (f(\pmb{X}))<n_i,$ 
for $i\in[m],$ meaning 
$$C_{k}(\mathcal{A},\pmb{v})=\left\{{\rm ev}_k(f(\pmb{X}))\colon f(\pmb{X})\in S_{< k},
\deg_{X_i} (f(\pmb{X}))<n_i \text{ for } i\in[m]\right\}.$$ Indeed, if $c \in C_{k}(\mathcal{A},\pmb{v})$, then 
there exists
$f(\pmb{X})\in S_{< k}$ such that ${\rm ev}_k(f(\pmb{X}))=c.$ By the division algorithm in $S$
\cite[Theorem 1.5.9]{AL}, there are $f_1(\pmb{X}),\ldots,f_m(\pmb{X}),r(\pmb{X})$ in $S$ such that
\[
f(\pmb{X})=\sum_{i=1}^mf_i(\pmb{X}) L_i(\pmb{X})+r(\pmb{X}),
\]
where $\deg_{X_i}(r(\pmb{X}))<n_i$ for $i= 1,\ldots,m,$ and
$\deg(r(\pmb{X}))\leq \deg(f(\pmb{X}))< k.$ Then  ${\rm ev}_k(f(\pmb{X}))={\rm ev}_k(r(\pmb{X}))$ and  
$C_{k}(\mathcal{A},\pmb{v})\subseteq
\left\{{\rm ev}(f(\pmb{X}))\colon f(\pmb{X})\in S_{< k},\deg_{X_i} (f(\pmb{X}))<n_i \text{ for } i\in[m]\right\}.$ Given this, moving forward, we make the assumption that $\deg_{X_i} (f(\pmb{X}))<n_i$.

Next we point out that the map ${\rm ev}_k$ is injective when $\deg_{X_i} (f)<n_i$ for $i\in[m].$ It is easy to see that $f(\pmb{a})=g(\pmb{a})$ for all $\pmb{a}\in \mathcal{A}$ implies $(f-g)(\pmb{X})\in I(\mathcal{A}).$ However,  
$\deg_{X_i} (f)<n_i $ and $\deg_{X_i} (g)<n_i$ for $i\in[m]$ forces 
$\deg_{X_i} (f-g)<n_i$ for $i\in[m]$.  As a result, $(f-g)(\pmb{X})=0$ meaning $f(\pmb{X})=g(\pmb{X})$.

For each element $\pmb{a}=\left(a_1,\ldots,a_m\right)$ of $\mathcal{A}$, define the polynomial
\begin{equation}\label{03-23-18}
L_{\pmb{a}}(\pmb{X}):=
\frac{L_1(X_1)}{(X_1-a_1)}\cdots
\frac{L_m(X_m)}{(X_m-a_m)}\in S.
\end{equation}
For $\pmb{b}=\left(b_1,\ldots,b_m\right)$, it is straightforward to check that $L_{\pmb{a}}(\pmb{b})= 0$ if and only if $\pmb{a}\neq \pmb{b}.$ In addition,  $L_{\pmb{a}}(\pmb{a})=L_1^\prime(a_1)\cdots L_m^\prime(a_m),$ where $L_i^\prime(X_i)$ denotes the formal derivative of $L_i(X_i).$ Writing $L_{\pmb{a}}(\pmb{a})$ in terms of the derivatives is convenient at times for computational purposes.

Given $c = \left( c_{\pmb{a_1}}, \dots, c_{\pmb{a_n}} \right)  \in C_{k}(\mathcal{A},\pmb{v}),$
there exists a polynomial $f(\pmb{X})$ in $S_{< k}$ such that ${\rm ev}_k(f(\pmb{X}))=c$
and $\deg_{X_i}(f(\pmb{X}))<n_i$ for $i\in[m].$ Define the polynomial
\[f_c(\pmb{X}):=\sum_{\pmb{a}\in\mathcal{A}}
\frac{L_{\pmb{a}}(\pmb{X})}{L_{\pmb{a}}(\pmb{a})}c_{\pmb{a}}.
\]
Then 
$f_c(\pmb{a_i})=c_{\pmb{a_i}}$ and
${\rm ev}_k(f_c)=c.$ By definition of $L_{\pmb{a}}$
we have $\deg_{X_i}(f_c(\pmb{X}))<n_i$ for $i\in[m].$ Based on the injectivity of ${\rm ev}_k$
when we restrict to polynomials with $\deg_{X_i}(f_c(\pmb{X}))<n_i$,
$f_c(\pmb{X})$ is the unique polynomial in $S_{< k}$ such that
${\rm ev}_k(f_c(\pmb{X}))=c$ and $\deg_{X_i}(f_c(\pmb{X}))<n_i$ for $i\in[m].$

Using these ideas we are almost ready to find the dual of a Cartesian code. Just one more result.

\begin{lemma}\label{02-07-17}
Let $k^\prime=\sum_{i = 1}^m (n_i - 1)-k+1$. Then 
$\dim(C_{k}(\mathcal{A},\pmb{v}))+\dim(C_{k^\prime}({\mathcal{A},\pmb{v}}))=
n_1\cdots n_m.$
\end{lemma}
\begin{proof}
Observe that the dimension of $C_{k}(\mathcal{A},\pmb{v})$ given in Theorem~\ref{01-28-17} {\rm (2)}
is the number of integer solutions of the following inequality
\begin{equation}\label{02-03-17}
x_1+\cdots+x_m\leq k-1,\quad \text{ where }\quad 0\leq x_i \leq n_i-1 \quad \text{ for } i\in [m]. 
\end{equation}
The number of integer solutions of the inequality
\begin{equation}\label{02-04-17}
x_1+\cdots+x_m> k-1,\quad \text{ where }\quad 0\leq x_i \leq n_i-1 \quad \text{ for } i\in [m],
\end{equation}
is the same that the number of integer solutions of the inequality
\[
x_1+\cdots+x_m < \sum_{i = 1}^m (n_i - 1)-k+1=k^\prime,\quad
\text{ where }\quad 0\leq x_i \leq n_i-1 \quad \text{ for } i\in [m],
\]
which is the dimension of $C_{k^\prime}({\mathcal{A},\pmb{v}}).$
As the total number of integer solutions of (\ref{02-03-17})
plus the total number of integer solutions of (\ref{02-04-17}) is
$n_1\cdots n_m,$ we obtain the result.
\end{proof}

We come to the main result of this section.
\begin{theorem}\label{03-19-18}
Let $\pmb{a_1},\ldots,\pmb{a_n}$ be points of the Cartesian set
$\mathcal{A}=A_1\times\cdots\times A_m.$
The dual of $C_{k}(\mathcal{A},\pmb{v})$ is 
$$C_{k}(\mathcal{A},\pmb{v})^\perp=C_{k^\prime}(\mathcal{A},\pmb{v}^\prime),$$
where 
 $k^\prime=\sum_{i = 1}^m (n_i - 1)-k+1$ and $\pmb{v}^\prime$ is given by $v_i^\prime:=(v_iL_{\pmb{a_i}}(\pmb{a_i}))^{-1}.$
\end{theorem}
\begin{proof}
Let $f(\pmb{X})$ be an element of $S_{< k}$ such that $\deg_{X_i}(f(\pmb{X}))<n_i$ for $i\in [m]$ and
let $g(\pmb{X})$ be an element of $S_{< k^\prime}$ such that $\deg_{X_i}(g(\pmb{X}))<n_i$ for $i\in [m].$
By the division algorithm in $S$ \cite[Theorem 1.5.9]{AL}, there are
$f_1(\pmb{X}),\ldots,f_m(\pmb{X}),r(\pmb{X})$ in $S$ such that
\[f(\pmb{X})g(\pmb{X})=\sum_{i=1}^mf_i(\pmb{X})L_i(X_i)+r(\pmb{X}),
\]
where $\deg_{X_i}(r(\pmb{X}))<n_i$ for $i= 1,\ldots,m,$ and
\begin{equation}\label{02-06-17}
\deg(r(\pmb{X}))\leq \deg((fg)(\pmb{X}))\leq \sum_{i = 1}^m (n_i - 1)-1.
\end{equation}
Observe that $r(\pmb{a})=(fg)(\pmb{a})$ for all $\pmb{a}\in \mathcal{A}.$  Then \begin{equation}\label{03-17-18}
r(\pmb{X})=\sum_{\pmb{a}\in\mathcal{A}}
\frac{L_{\pmb{a}}(\pmb{X})}{L_{\pmb{a}}(\pmb{a})}
(fg){(\pmb{a})}.
\end{equation}
The coefficient of the monomial of degree $\sum_{i = 1}^m (n_i - 1)$ on the
right-hand side of \eqref{03-17-18} is given by:
\begin{eqnarray}
\sum_{\pmb{a}\in\mathcal{A}}
\frac{(fg){(\pmb{a})}}{L_{\pmb{a}}(\pmb{a})}=
\sum_{\pmb{a}\in\mathcal{A}}
\frac{f{(\pmb{a})}g{(\pmb{a})}}{L_{\pmb{a}}(\pmb{a})}=
\sum_{\pmb{a}\in\mathcal{A}}
\frac{v_{\pmb{a}}f{(\pmb{a})}g{(\pmb{a})}}{v_{\pmb{a}}L_{\pmb{a}}(\pmb{a})}=\\
=\left(v_{\pmb{a_1}}f(\pmb{a_1}),\ldots,v_{\pmb{a_n}}f(\pmb{a_n})\right)\cdot
\left(\frac{g(\pmb{a_1})}{v_{\pmb{a_1}}L_{\pmb{a_1}}(\pmb{a_1})},\ldots,
\frac{g(\pmb{a_n})}{v_{\pmb{a_n}}L_{\pmb{a_n}}(\pmb{a_n})}\right).\label{03-18-18}
\end{eqnarray}
By \eqref{02-06-17} $\deg(r(\pmb{X}))<\sum_{i = 1}^m (n_i - 1),$ so
the coefficient of the monomial of degree $\sum_{i = 1}^m (n_i - 1)$ on the
left-side of \eqref{03-17-18} is $0.$ Thus the dot product shown on
\eqref{03-18-18} is $0.$
As the left-hand side of the dot product given in \eqref{03-17-18} is an arbitrary element of
$C_{k}(\mathcal{A},\pmb{v})$, and
right-hand side of the dot product of Equation~\eqref{03-17-18} is an arbitrary element of
$C_{k^\prime}(\mathcal{A},\pmb{v}^\prime)$,  the proof is complete because
$\dim(C_{k^\prime}(\mathcal{A},\pmb{v}^\prime))=\dim(C_{k^\prime}(\mathcal{A},\pmb{v}))=\dim(C_k(\mathcal{A},\pmb{v})^\perp)$
where the last equality follows from 
Lemma~\ref{02-07-17}.
\end{proof}

\section{Finding LCD codes from Cartesian codes} \label{LCD}
In this section, we determine which Cartesian codes $C_{k}(\mathcal{A},\pmb{v})$, where $\mathcal{A}:=A_1\times\cdots\times A_m\subseteq K^{m}$, 
are LCD. As a result, a number of explicit constructions for LCD codes are found. 
\subsection{Generalized Reed-Solomon codes (i.e., the case $m=1$)}
We start with the case when $m=1$, meaning $\mathcal{A}=A_1:=\left\{a_1,\ldots,a_n \right\} \subseteq K$, so $n_1=n.$
Observe that in this case the Cartesian code
$C_{k}(A,\pmb{v})$ is the {\it generalized Reed-Solomon code} $GRS_{k}(A,\pmb{v}),$
which is given by
\[GRS_{k}(A,\pmb{v}):=\left\{\left(v_{1}f(a_1),\ldots,v_{n}f(a_n)\right)\mid 
f(X)\in K[X], \deg f(X) < k \right\}.\]
Recall 
$L(X)=\prod_{\substack{a\in A}}\left(X-a\right)$
and
$L_a(X)=\frac{L(X)}{(X-a)}$ for each element $a\in A.$
By \eqref{03-19-18},
\begin{align*}
C_{k}(A,\pmb{v})^\perp&=\left\{
\left(\frac{g({a_1})}{v_{{1}}L_{{a_1}}({a_1})},\ldots,
\frac{g({a_n})}{v_{{n}}L_{{a_n}}({a_n})}\right)\mid 
g(X)\in K[X], \deg g(X) < n-k
\right\}\\
&=C_{n-k}(A,\pmb{v}^\prime)=GRS_{n-k}(A,\pmb{v}^\prime).
\end{align*}
We are interested in finding conditions on $A$ and $\pmb{v}$ such that
$C_{k}(A,\pmb{v})$ is LCD. Observe that the Cartesian code $C_{k}(A,\pmb{v})$
is not LCD if and only if there are nonzero polynomials $f(X)$ and $g(X)$ such that $\deg(f(X))< k$,
$\deg(g(X))< n-k$ and
\begin{equation*}
\left(v_{1}f(a_1),\ldots,v_{n}f(a_n)\right)=\left(\frac{g({a_1})}{v_{1}L_{{a_1}}({a_1})},\ldots,
\frac{g({a_n})}{v_{n}L_{{a_n}}({a_n})}\right).
\end{equation*}
This holds if and only if
\begin{equation}\label{03-21-18}
v_{i}^2L_{a_i}(a_i) f(a_i)=g(a_i),\quad \text{ for all } a_i\in A.
\end{equation}
Define the polynomial associated to $C_{k}(A,\pmb{v})$  by
\begin{equation}\label{03-22-18}
H(X):=\sum_{{a_i}\in A}v_{i}^2 L_{a_i}(X).
\end{equation}
Notice that $H(a_i)=v_i^2L_{a_i}(a_i)$ for all $a_i\in A$ and $\deg(H) <
n$. Moreover, $H(X)$ and $L(X)$ are coprime in $K[X]$. To see this,
observe that $L(a)=0$  if and only if $a\in A$ whereas $H(a)\neq 0$ if
$a\in A$. Since $L(X)$ factors into linear terms over $K$, $H(X)$ and $L(X)$ have no nonconstants common factors. Then we have the following result.

\rmv{
\begin{lemma}\label{03-25-18}
The polynomial $H(X)$ as defined in Equation~\eqref{03-22-18} has the following properties:
\begin{itemize}
\item[\rm (i)] $H(a)=v_{a}^2L_{1}^\prime(a),$ for all $a\in A.$
\item[\rm (ii)] $\deg(H) < n.$
\item[\rm (iii)] $H(X)$ and $L(X)$ are coprime in $K[X].$
\end{itemize}
\end{lemma}
\begin{proof}

{\rm (i)} This is a consequence of Lemma~\ref{02-02-17} (i). {\rm (ii)} It is true
because by definition of $L_a(X)$ (Equation~\eqref{03-23-18}) $\deg (H)\leq\deg(L_a)< n.$
{\rm (iii)} Let $k$ be an element of $K.$ By definition of $L(X)$ (Equation~\eqref{03-24-18}),
$L(k)=0$ if and only if $k\in A,$ but $H(k)\neq 0$ for all $k\in A$ by (i).
\end{proof}}
\begin{proposition}\label{03-27-18}
The Cartesian code  $C_{k}(A,\pmb{v})$ is  LCD if and only if for all nonzero
polynomials  $f(X), g(X)\in K[X]$ with $\deg(f(X))< k$ and $\deg(g(X))< n-k$ we have
\begin{equation*}
H(X)f(X) - g(X) \not\in \langle L(X)\rangle
\end{equation*}
where $H(X)$ is defined in \eqref{03-22-18}.
\end{proposition}
\begin{proof}
Equation~\eqref{03-21-18} holds
if and only if $L(x)$ divides $H(X)f(X)-g(X)$.
\end{proof}

As $H(X)$ and $L(X)$ are coprime, by the Extended Euclidean Algorithm \cite[Chapter 3]{vo13}, there exists a natural number $t$, polynomials $g_i(X),\ h_i(X),\ f_i(X)\in K[X]$ for
$i\in\left\{0,\ldots,t+1\right\}$ and polynomials $q_i(X)\in K[X]$ for
$i\in \{1,\ldots,t\}$ such that 
\begin{align} 
  &g_0=L,\ g_1 =H,\ h_0=f_1=1,\ h_1=f_0=0\nonumber\\
  &g_{i-1}=q_{i}g_{i}+g_{i+1} \mbox{ where }  \deg g_{i+1}<\deg g_i &
                                                                      \forall \ i\in \{1,\dots, t\} \nonumber\\
  &g_{i}=h_i L+ f_i H & \forall \ i\in \{0,\dots, t\}\label{gihifi}\\ 
  &\deg f_i = \deg L - \deg g_{i-1} = n - \deg g_{i-1}& \forall \ i\in
                                                        \{1,\dots, t\}
  \label{degfi}\\
  &g_{t+1}=1. \nonumber
\end{align}

The following is the basis of our main results of this section.

\begin{proposition}\label{03-31-18}
 Let $C_{k}(A,\pmb{v})$ be a Cartesian code and
  $g_0(X),\ldots,g_{t+1}(X)$ be the remainders from the Extended
  Euclidean Algorithm applied to polynomials
  $L(X)=\prod_{\substack{a_i\in A}}\left(X-a_i\right)$ and
  $H(X)=\sum_{{a_i}\in A}v_{i}^2 L_{a_i}(X)$.  Then, $C_{k}(A,\pmb{v})$ is
  LCD if and only if
for all $i\in\left\{1,\ldots,t+1\right\}$, 
\[ \deg(g_{i-1}(X)) \leq n-k \quad {\mbox or} \quad \deg(g_i(X)) \geq n-k.\]
\end{proposition}

\begin{proof} We prove both implications via the contrapositives. 

($\Rightarrow$)
Assume there is $i\in\left\{1,\ldots,t\right\}$ such that
$\deg(g_i(X))< n-k< \deg(g_{i-1}(X)).$
Then by  \eqref{gihifi} and \eqref{degfi}, there are $f_i(X),g_i(X)$ and
$h_i(X)$ in $K[X]$ such that $L(X)h_i(X)+H(X)f_i(X)=g_i(X),$ and 
$\deg(g_i(X))<n-k$ and $\deg(f_i(X))=n-\deg(g_{i-1}(X))<k.$ By
Proposition~\ref{03-27-18}, $C_{k}(A,\pmb{v})$ is not LCD.  

($\Leftarrow$) Assume $C_{k}(A,\pmb{v})$ is not LCD. By
Proposition~\ref{03-27-18}, there are polynomials $f(X), g(X)$ and
$h(X)$ in $K[X]$ such that $\deg(f(X))< k,$ $\deg{g(X)}< n-k$ and
\begin{equation}\label{03-28-18}
  L(X)h(X)+H(X)f(X)=g(X).
\end{equation}
Let $g_i(X)$ be the remainder such that $\deg(g_i(X))\leq\deg(g(X))$ and
$\deg(g_{i-1}(X))>\deg(g(X)).$ Observe $\deg(g_i(X))\leq\deg(g(X))<n-k,$ which
means now we just need to prove $\deg(g_{i-1}(X))>n-k.$

Combining \eqref{03-28-18} with \eqref{gihifi}, we obtain that
\begin{equation*}
  L(X)\left(h(X)g_i(X)-h_i(X)g(X)\right)+
  H(X)\left(f(X)g_i(X)-f_i(X)g(X)\right)=0.
\end{equation*}
Because $L(X)$ and $H(X)$ are coprime,
$L(X) \mid f(X)g_i(X)-f_i(X)g(X).$ Moreover it holds that
\begin{align*}
\deg((fg_i)(X))&=\deg(f(X))+\deg(g_i(X))\leq\deg(f(X))+\deg(g(X))<n, \mbox{ and}\\
\deg((f_ig)(X))&=\deg(f_i(X))+\deg(g(X))=\\
&=n - \deg(g_{i-1}(X))+\deg(g(X))<n-\deg(g(X))+\deg(g(X))=n.
\end{align*}
Then
\begin{equation*}
  f(X)g_i(X)=f_i(X)g(X),
\end{equation*}
which implies $\deg(f_i(X))=\deg(f(X))+\deg(g_i(X))-\deg(g(X)).$ Then by \eqref{degfi}
\[\deg(g_{i-1}(X))=n-\deg(f_i(X))=n-\deg(f(X))-\deg(g_i(X))+\deg(g(X))> n-k,\]
which completes the proof.

\end{proof}
The following theorem is the main result of this section and a
corollary of Proposition \ref{03-31-18}.
\begin{theorem}\label{04-01-18}
  Let $C_{k}(A,\pmb{v})$ be a Cartesian code and
  $g_0(X),\ldots,g_{t+1}(X)$ be the remainders from the Extended
  Euclidean Algorithm applied to polynomials
  $L(X)=\prod_{\substack{a_i\in A}}\left(X-a_i\right)$ and
  $H(X)=\sum_{{a_i}\in A}v_{i}^2 L_{a_i}(X)$.  Then, $C_{k}(A,\pmb{v})$ is
  LCD if and only if
\[n-k\in \left\{n-1,n-2,\ldots,\deg(g_1(X)),\deg(g_2(X)),\ldots,\deg(g_{t+1}(X))\right\}.\]
\end{theorem}
\subsection{Affine Cartesian codes on $m>1$ components}
Let $\pmb{a_1},\ldots,\pmb{a_n}$ be the points of the Cartesian set
$\mathcal{A}=A_1\times\cdots\times A_m.$
Now we are ready to determine whether a Cartesian code
$C_{k}(\mathcal{A},\pmb{v})$ is LCD. 
By \eqref{03-19-18}, the dual of $C_{k}(\mathcal{A},\pmb{v})$ is given by
$$C_{k}(\mathcal{A},\pmb{v})^\perp=C_{k^\prime}(\mathcal{A},\pmb{v}^\prime),$$
where $k^\prime=\sum_{i = 1}^m (n_i - 1)-k+1$ and $\pmb{v}^\prime$ is defined by
$v_{i}^\prime:=v_{i}^{-1}L_{{\pmb{a_i}}}({\pmb{a_i}})^{-1}.$
Thus the Cartesian code
$C_{k}(\mathcal{A},\pmb{v})$ is not LCD if and only if there are nonzero polynomials
$f(\pmb{X})\in S_{< k}$ and $g(\pmb{X})\in S_{< k^\prime}$ such that
\begin{equation}\label{04-03-18}
\left(v_{1}f(\pmb{a_1}),\ldots,v_{n}f(\pmb{a_n})\right)=
\left(\frac{g({\pmb{a_1}})}{v_{1}L_{{\pmb{a_1}}}({\pmb{a_1}})},\ldots,
\frac{g({\pmb{a_n}})}{v_{n}L_{{\pmb{a_n}}}({\pmb{a_n}})}\right).
\end{equation}
Equation~\eqref{04-03-18} holds if and only if
\begin{equation}\label{04-04-18}
v_{i}^2L_{\pmb{a_i}}(\pmb{a_i}) f(\pmb{a_i})=g(\pmb{a_i}),\quad \text{ for all } \pmb{a_i}\in \mathcal{A}.
\end{equation}
The {\it polynomial associated} to $C_{k}(\mathcal{A},\pmb{v})$ is defined by
\begin{equation}\label{04-05-18}
H(\pmb{X}):=\sum_{\pmb{a_i}\in \mathcal{A}}
v_{i}^2{L_{\pmb{a_i}}(\pmb{X})}.
\end{equation}
Notice that \begin{enumerate} \item for all
$\pmb{a_i} \in \mathcal{A},$
 $H(\pmb{a_i})=v_{i}^2L_{\pmb{a_i}}(\pmb{a_i})$
and  \item $\deg_{X_i}(H(\pmb{X})) < n_i$ for all $i\in[n]$. 
\end{enumerate}
Moreover, if $G(\pmb{X})$ is an element of $S$ that satisfies 
1. and 2., then $G(\pmb{X})=H(\pmb{X}).$
\rmv{
\begin{lemma}\label{04-06-18}
$H(\pmb{X}),$ defined on Equation~\eqref{04-05-18} has the following properties:
\begin{itemize}
\item[\rm (i)] $H(\pmb{a})=v_{\pmb{a}}^2L_{\pmb{a}}(\pmb{a})=
v_{\pmb{a}}^2L_1^\prime(a_1)\cdots L_m^\prime(a_m),$ for all
$\pmb{a}=(a_1,\ldots,a_m)\in \mathcal{A}.$
\item[\rm (ii)] $\deg_{X_i}(H) < n_i,$ for all $i\in[n].$
\item[\rm (iii)] If $G(\pmb{X})$ is an element of $S$ that satisfies 
(i) and (ii), then $G(\pmb{X})=H(\pmb{X}).$
\end{itemize}
\end{lemma}
\begin{proof}
{\rm (i)} This is a consequence of Lemma~\ref{02-02-17} (i) and (ii). {\rm (ii)} It is true
because by definition of $L_{\pmb{a}}(\pmb{X})$ (Equation~\eqref{03-23-18}),
$\deg_{X_i}(H)\leq\deg(L_i)-1< n_i.$
{\rm (iii)} As $H(\pmb{X})$ and $G(\pmb{X})$ satisfies (i), then
$(G-H)(\pmb{X})\in I(\mathcal{A}).$ By (ii) and Lemma~\ref{01-29-17}
the result follows.
\end{proof}}
We have the following characterization for LCD codes.

\begin{proposition}\label{04-07-18} The Cartesian code
$C_{k}(\mathcal{A},\pmb{v})$ is LCD if and only if for all  nonzero polynomials
$f(\pmb{X}), g(\pmb{X}) \in S$ with $\deg(f)< k,$ and $\deg(g)\leq \sum_{i = 1}^m (n_i - 1)-k$
we have
\begin{equation*}
H(\pmb{X})f(\pmb{X}) - g(\pmb{X}) \not\in  I(\mathcal{A}),
\end{equation*}
where $H(\pmb{X})$ is the polynomial associated to $C_{k}(\mathcal{A},\pmb{v})$ defined in
\eqref{04-05-18}
\end{proposition}
\begin{proof}
Equation~\eqref{04-04-18} holds
if and only if $(Hf-g)(\pmb{X})\in I(\mathcal{A}).$
\end{proof}

\rmv{
Then we have the following characterization for non LCD codes.
\begin{proposition}\label{04-07-18}
$C_{k}(\mathcal{A},\pmb{v})$ is not LCD if and only if there are polynomials
$f(\pmb{X}), g(\pmb{X}),$ $h(\pmb{X})$ and $L(\pmb{X})$
in $S$ such that $\deg(f)\leq k,$ $\deg(g)\leq \sum_{i = 1}^m (n_i - 1)-k-1,$
$L(\pmb{X})\in I(\mathcal{A})$ and
\begin{equation*}
L(\pmb{X})h(\pmb{X})+H(\pmb{X})f(\pmb{X})=g(\pmb{X}),
\end{equation*}
where $H(\pmb{X})$ is the polynomial associated to $C_{k}(\mathcal{A},\pmb{v})$ defined on Equation~\ref{04-05-18}.
\end{proposition}
\begin{proof}
Equation~\eqref{04-04-18} holds
if and only if $(Hf-g)(\pmb{X})\in I(\mathcal{A}).$
\end{proof}}

Next, we focus on a special family of Cartesian codes. 
\begin{definition}\label{04-08-18}
For $i\in [m]$,  write $A_i:=\left\{a_{i1},\ldots,a_{in_i} \right\} \subseteq K$ and let $\pmb{v_i}:= \left(v_{i1},\ldots,v_{in_i} \right) \in (K^*)^n$. The {Cartesian vector}
$\pmb{v_1}\times\cdots\times\pmb{v_m} \in K^n$
 is defined as a vector of length $n_1\cdots n_m$
 with 
$$\left(\pmb{v_1}\times\cdots\times\pmb{v_m}\right)_{\pmb{a}}:=
{v}_{1j_1}\cdots{v}_{mj_m} \quad \text{ where } \quad \pmb{a}=(a_{1j_1},\ldots,a_{mj_m}).$$
\end{definition}

In a few words, the following results says that if the Cartesian code
$C_{k}(\mathcal{A},\pmb{v_1}\times\cdots\times\pmb{v_m})$ is not LCD, then
one of its components $C_{t_i}(A_i, \pmb{v_i})$ is not LCD, for some $t_i\leq \min\left\{k,n_i-1\right\}.$
\begin{theorem}\label{05-03-18}
If for all $i\in \left\{ 1,\ldots, m\right\}$ we have that
$C_{t_i}(A_i,\pmb{v_i})$ is LCD for all $t_i< \min\left\{k,n_i\right\}$, then  
$C_{k}(\mathcal{A},\pmb{v_1}\times\cdots\times\pmb{v_m})$ is $LCD$.
\end{theorem}
\rmv{\begin{theorem}
If $C_{k}(\mathcal{A},\pmb{v}^1\times\cdots\times\pmb{v}^m)$ is not $LCD$ then there is
 $i\in[m]$ and a natural number $t_i\leq \min\left\{k,n_i-1\right\}$ such that
$C_{t_i}(A_i,\pmb{v}^i)$ is not LCD.
\end{theorem}}
\begin{proof}
Assume $C_{k}(\mathcal{A},\pmb{v_1}\times\cdots\times\pmb{v_m})$ is not LCD.
By Proposition~\ref{04-07-18} there are polynomials
$f(\pmb{X}), g(\pmb{X}),$ $h(\pmb{X})$ and $L(\pmb{X})$
in $S$ such that $\deg(f)< k,$ $\deg(g)\leq \sum_{i = 1}^m (n_i - 1)-k,$
$L(\pmb{X})\in I(\mathcal{A})$ and
\begin{equation*}
L(\pmb{X})h(\pmb{X})+H(\pmb{X})f(\pmb{X})=g(\pmb{X}),
\end{equation*}
where $H(\pmb{X})$ is the polynomial associated to
$C_{k}(\mathcal{A},\pmb{v}_1\times\cdots\times\pmb{v}_m)$ defined on Equation~\ref{04-05-18}.
As $\deg_{X_i}(f(\pmb{X}))<n_i,$ by
{\rm Combinatorial Nullstellensatz
\cite[Theorem~1.2]{alon-cn}},
there are $a_1\in A_1,\ldots, a_m\in A_m$ such that
$f(a_1,\ldots,a_m)\neq 0,$ which implies
$g(a_1,\ldots,a_m)\neq 0.$ These $a_i$'s give
the following $m$ equations:
\begin{eqnarray*}
L_1(X_1)h(X_1,\ldots,a_m)+H(X_1,\ldots,a_m)f(X_1,\ldots,a_m)&=&
g(X_1\ldots,a_m)\\
L_1(X_2)h(a_1,X_2,\ldots,a_m)+H(a_1,X_2,\ldots,a_m)f(a_1,X_2,\ldots,a_m)&=&
g(a_1,X_2,\ldots,a_m)\\
&\vdots& \\
L_m(X_m)h(a_1,\ldots,X_m)+H(a_1,\ldots,X_m)f(a_1,\ldots,X_m)&=&
g(a_1,\ldots,X_m).
\end{eqnarray*}
Let $t_i:=\deg_{X_i}(f(\pmb{X})).$
Observe
\begin{align*}
\sum_{i=1}^m\deg_{X_i}(g(a_1,\ldots,X_i,\ldots,a_m))&\leq
\deg(g)\leq \sum_{i = 1}^m (n_i - 1)-k\\&< \sum_{i = 1}^m (n_i - 1)-\sum_{i = 1}^m t_i=
\sum_{i = 1}^m (n_i - t_i - 1).
\end{align*}
Thus there is $i\in\left\{1,\ldots,m\right\}$ such that 
$$\deg_{X_i}(g(a_1,\ldots,a_{i-1},X_i,a_{i+1},\ldots,a_m))<n_i-t_i-1.$$
Observe that  
$H(a_1,\ldots,a_{i-1},X_i,a_{i+1},\ldots,a_m)=
\ell v_i^2L_{i}(X_i)$ where $\ell$ is a nonzero constant. Thus the $i$-th equation
implies by Proposition~\ref{03-27-18} that the Cartesian code
$C_{t_i}\left(A_i,\pmb{v_i}\right)$ is not LCD.
\end{proof}
\begin{theorem}
If $k<\min\{ n_i\mid i\in [n]\}$ and $C_{k}(\mathcal{A},\pmb{v_1}\times\cdots\times\pmb{v_m})$ is LCD, then 
$C_{k}(A_i,\pmb{v_i})$ is LCD
for all $i\in \left\{ 1,\ldots, m\right\}.$
\end{theorem}
\begin{proof}
If there is $i\in \left\{ 1,\ldots, m\right\}$ such that $C_{k}(A_i,\pmb{v_i})$ is not LCD,
by Proposition~\ref{03-27-18},
there are polynomials $f(X_i), g(X_i)$ and $h(X_i)$
in $K[K_i]$ such that $\deg(f(X_i))< k,$ $\deg{g(X_i)}< n_i-k$ and
\begin{equation*}
L_i(X_i)h(X_i)+H_i(X_i)f(X_i)=g(X_i).
\end{equation*}
For $j\neq i,$ let $H_j(X_j)$ be the polynomial associated to
$C_{k}(A_j,\pmb{v_j})$ defined in Equation~\ref{03-22-18}.
The following equations holds.
\begin{equation*}
H(X_1)\cdots H_m(X_m)f(X_i)=
g(X_i)H(X_1)\cdots H_{i-1}(X_{i-1})H_{i+1}(X_{i+1})\cdots H_m(X_m)\text{ mod } L_i(X_i).
\end{equation*}
Observe $\deg(gH\cdots H_{i-1}H_{i+1}\cdots H_m)<\sum_{i = 1}^m (n_i - 1)-k+1.$
As $H(X_1)\cdots H(X_m)$ is the polynomial associated to 
$C_{k}(\mathcal{A},\pmb{v_1}\times\cdots\times\pmb{v_m})$,
the thesis follows from Proposition~\ref{04-07-18}.
\end{proof}
We come to one of the main results of this section.
\begin{theorem}
If at least one of the Cartesian codes $C_{t_1}(A_1,\pmb{v_1}),\ldots,C_{t_m}(A_m,\pmb{v_m})$ is not LCD,
then
$C_{t_1+\cdots+t_m}(\mathcal{A},\pmb{v_1}\times\cdots\times\pmb{v_m})$ is not $LCD.$
\end{theorem}
\begin{proof}
Assume for all $i\in\left\{1,\ldots,m\right\}$
$C_{t_i}(A_i,\pmb{v_i})$ is not LCD. By Proposition~\ref{03-27-18}
there are polynomials $f_i(X_i), g_i(X_i)$ and $h_i(X_i)$
in $K[X_i]$ such that $\deg(f_i(X_i))< t_i,$ $\deg(g_i(X_i))< n_i-t_i$ and
\begin{equation}\label{04-09-18}
L_i(X_i)h_i(X_i)+H_i(X_i)f_i(X_i)=g_i(X_i) \qquad \text{ for } \qquad i\in\left\{1,\ldots,m\right\},
\end{equation}
where $H_i(X_i)$ is the polynomial associated to $C_{t_i}(A_i,\pmb{v}_i)$
defined on Equation~\ref{03-22-18}.
Multiplying the all $m$ equations from Equation~\eqref{04-09-18} we obtain an expresion
of the form
\begin{equation*}
L(\pmb{X})G+H(\pmb{X})f_1(X_1)\cdots f_m(X_m)=g_1(X_1)\cdots g_m(X_m),
\end{equation*}
where $L(\pmb{X})\in I(\mathcal{A}),$
$H(\pmb{X}):=H(X_1)\cdots H_m(X_m)$ is the polynomial
associated to\\ $C_{t_1+\cdots+t_m}(\mathcal{A},\pmb{v_1}\times\cdots\times\pmb{v_m}),$
$\deg(f_1(X_1)\cdots f_m(X_m))< t_1+\cdots+t_m$ and $\deg(g_1(X_1)\cdots g_m(X_m)) \leq
\sum_{i=1}^m(n_i-1)-(t_1+\cdots+t_m).$ Proposition~\ref{04-07-18} gives the result.
\end{proof}

\section{Examples} \label{ex}

In this section, we record some examples of our results. First, we show that the LCD property depends of the scalars and also of the evaluations points. 

\begin{example}
Let $K:=\mathbb{F}_{7},$ $A:=\left\{0,1,2\right\}$ and 
$\pmb{v} :=  \left(1,1,1 \right).$
A generator matrix of the code $C_{2}(\mathcal{A},\pmb{v})$
is given by
$$G:=\left(\begin{array}{ccc}1 & 1 & 1 \\0 & 1 & 2\end{array}\right).$$
As $GG^T=\left(\begin{array}{cc}3 & 3 \\3 & 0\end{array}\right),$ which
is nonsingular, then the code $C_{2}(\mathcal{A},\pmb{v})$ is LCD.
\end{example}
\begin{example}
Let $K:=\mathbb{F}_{7},$ $A:=\left\{0,1,2\right\}$ and 
$\pmb{v} :=  \left(1,1,2 \right).$
A generator matrix of the code $C_{2}(\mathcal{A},\pmb{v})$
is given by
$$G:=\left(\begin{array}{ccc}1 & 1 & 2 \\0 & 1 & 4\end{array}\right).$$
As $GG^T=\left(\begin{array}{cc}6 & 2 \\2 & 3\end{array}\right),$ which
is singular, then the code $C_{2}(\mathcal{A},\pmb{v})$ is not LCD.
\end{example}

\begin{example}
Let $K:=\mathbb{F}_{7},$ $A:=\left\{0,1,3\right\}$ and 
$\pmb{v} :=  \left(1,1,1 \right).$
A generator matrix of the code $C_{2}(\mathcal{A},\pmb{v})$
is given by
$$G:=\left(\begin{array}{ccc}1 & 1 & 1 \\0 & 1 & 3\end{array}\right).$$
As $GG^T=\left(\begin{array}{cc}3 & 4 \\4 & 3\end{array}\right),$ which
is singular, then the code $C_{2}(\mathcal{A},\pmb{v})$ is not LCD.
\end{example}

Next, we find some LCD codes which are also MDS. 

\begin{example}\label{05-17-18}
Let $K:=\mathbb{F}_{13}$ and $A:=\left\{0,2,3,5,6,8,10,11\right\}.$
Then the degrees of the remainders are $0,3,4,5,6$ and $7$. Thus, 
the generalized Reed-Solomon code $GRS_{k}(A,\pmb{1})$ is LCD
if and only if $k\in \left\{ 1, 2, 3, 4, 5, 8 \right\}.$
\end{example}
\begin{example}\label{05-18-18}
Using $A$ as in the previous example but with $K:=\mathbb{F}_{17},$
we obtain that the degrees of the remainders are $0,..,7.$ Thus, 
the generalized Reed-Solomon code $GRS_{k}(A,\pmb{1})$ is always LCD.
Of course $1\leq k \leq 8.$
\end{example}

These examples demonstrate that one can find LCD MDS codes beyond those found by Chen and Liu \cite{ChenLiu}. 
Their result, recorded below for easy reference, shows that assuming some conditions over $q,k$ and $n,$
it is always possible to find LCD MDS codes. Hence, our results are complementary to theirs in some sense.

\begin{theorem}\cite[Theorem I.1]{ChenLiu}
Let $q > 3$ be an odd prime power. Let $n > 1$ and $k$ be positive integers with
$1 < k \leq \lfloor \frac{n}{2} \rfloor .$
Then there exists a $q$-ary $[n, k]$ LCD MDS code whenever one of the following conditions holds.
\item $(1)$ $n = q + 1.$
\item $(2)$ $n > 1$ is a divisor of $q -1.$
\item $(3)$ $q = p^e$ and $n = p^\ell$, where $p$ is prime and $1\leq \ell \leq e.$
\item $(4)$ $n < q$ and $n + k \geq  q + 1.$
\item $(5)$ $n < q$ and $2n - k < q \leq 2n.$
\end{theorem}

Next, we mention some instances of LCD codes which are Cartesian codes on more than one component.
\begin{example}
Let $K:=\mathbb{F}_{13}$ and $A_1:=\cdots:=A_m:=\left\{0,2,3,5,6,8,10,11\right\}.$
By Example~\ref{05-17-18} we have that
every
$GRS_{k}(A_i,\pmb{1})$ is LCD for all $k\in\left\{1,\ldots,5\right\}$.
Using Theorem~\ref{05-03-18} we have
$GRS_{k}(A_1\times \cdots\times A_m,\pmb{1})$ is LCD for all $1\leq k \leq 5.$
\end{example}

\begin{example}
Let $K:=\mathbb{F}_{17}$ and $A_1:=\cdots:=A_m:=\left\{0,2,3,5,6,8,10,11\right\}.$
By Example~\ref{05-18-18} we have that
every
$GRS_{k}(A_i,\pmb{1})$ is LCD for all $k\in\left\{1,\ldots,8\right\}$.
Using Theorem~\ref{05-03-18} we have
$GRS_{k}(A_1\times \cdots\times A_m,\pmb{1})$ is LCD for all $1\leq k \leq 8.$
\end{example}

\section{Conclusion} \label{conclusion} 

In this paper, we studied affine Cartesian codes which are LCD codes. In doing so, it was necessary to consider generalized affine Cartesian codes, because they arise as duals of the affine Cartesian codes. Some results on this more general family of codes are included. We provide a characterization of which generalized Reed-Solomon codes are LCD, regardless of the characteristic of the ambient field, as well of which affine Cartesian codes are LCD. In addition, we explore certain families of generalized affine Cartesian codes which inherit the LCD property from their factors. This work allows us to find additional instances and explicit constructions of LCD codes.

\end{document}